\documentclass[journal]{IEEEtran}

\usepackage{amsmath}
\usepackage{amsfonts}
\usepackage{amssymb}
\usepackage{amsthm}

\newtheorem{mylemma}{Lemma}

\newtheorem{mytheorem}{Theorem}
\usepackage{graphics}
\usepackage{graphicx}

\usepackage{comment}
\usepackage{cite}
\usepackage{newtxmath}
\usepackage{newtxtext}
\usepackage{bm} 
\usepackage{enumerate}
\usepackage[bookmarks,colorlinks]{hyperref}

\usepackage{clipboard}
\newclipboard{myclipboard}

\def\A{\bm{A}}
\def\B{\bm{B}}

\def\T{\bm{T}}
\def\D{\bm{D}}

\def\X{\bm{X}}
\def\Y{\bm{Y}}
\def\x{\bm{x}}
\def\y{\bm{y}}
\def\z{\bm{z}}

\def\c{\bm{c}}
\def\E{\bm{E}}
\def\G{\bm{G}}
\def\v{\bm{v}}
\def\U{\bm{U}}

\usepackage{algorithm}  
\usepackage{algorithmicx} 
\usepackage{algpseudocode} 

\usepackage{color}
\usepackage{soul}

\begin{document}

    \title{Unambiguous Delay-Doppler Recovery from Random Phase Coded Pulses}

    \author{ Xiang Liu, Deborah Cohen, Tianyao Huang, Yimin Liu, Yonina C. Eldar
        
        \thanks{This project is funded by the National Natural Science Foundation of China under Grants No. 61801258,  European Union’s Horizon 2020 research and innovation program under grant agreement No. 646804-ERC-COG-BNYQ, from the Air Force Office of Scientific Research under grant No. FA9550-18-1-0208. Deborah Cohen is grateful to the Azrieli Foundation for the award of an Azrieli Fellowship.}
        
        \thanks{Y. C. Eldar is with the  Faculty of Mathematics and Computer Science, Weizmann Institute of Science, Rehovot,  Israel (e-mail: yonina.eldar@weizmann.ac.il).}
        
        
        \thanks{D. Cohen was with the Faculty of Electrical Engineering, Technion - Israel Institute of Technology, Haifa, Israel, when performing part of this work, and is now with Google Research, Tel Aviv, Israel. }
        
        \thanks{T. Huang, Y. Liu, and X. Liu are with the  Department of  Electronic Engineering, Tsinghua University, Beijing, China.}
        
    }

    \maketitle
    
    \begin{abstract}
        Pulse Doppler radars suffer from range-Doppler ambiguity that translates into a trade-off between maximal unambiguous range and velocity.
        Several techniques, like the multiple PRFs (MPRF) method, have been proposed to mitigate this problem. 
        The drawback of the MPRF method is that the received samples are not processed jointly, decreasing signal to noise ratio (SNR).
        To overcome the drawbacks of MPRF, we employ a random pulse phase coding approach to increase the unambiguous range region while preserving the unambiguous Doppler region.
        Our method encodes each pulse with a random phase, varying from pulse to pulse, and then processes the received samples jointly to resolve the range ambiguity.
        This technique increases the SNR through joint processing without the parameter matching procedures required in  MPRF.
        The recovery algorithm is designed based on orthogonal matching pursuit so that it can be directly applied to either Nyquist or sub-Nyquist samples.
        The unambiguous delay-Doppler recovery condition is derived using compressed sensing theory in noiseless settings. 
        In particular, an upper bound on the number of targets is given, with respect to the number of samples in each pulse repetition interval and the number of transmit pulses.
        Simulations show that in both regimes of Nyquist and sub-Nyquist samples our method outperforms the popular MPRF approach in terms of hit rate.
    \end{abstract}

    \IEEEpeerreviewmaketitle

    \section{Introduction} \label{section1}
    Pulse Doppler radars, which simultaneously estimate targets' range and velocity, are widely used for both civilian and military purposes, including meteorological applications \cite{doppler2014, polarimetric2014}, surveillance and tracking systems \cite{599236}.
    However, such systems suffer from the so-called ``range-Doppler ambiguity dilemma'' \cite{doppler2014, polarimetric2014}.
    For a certain pulse repetition interval (PRI) $T_r$, the maximum unambiguous range is $R_{\max} = c T_r / 2$, where $c$ is the propagation velocity, and the maximum unambiguous velocity is $V_{\max} = \lambda / (4 T_r)$, where $\lambda$ is the radar wavelength.
    This fundamental problem creates a trade-off between range and velocity ambiguity and limits their product to $R _ {\max} V _ {\max} = c \lambda / 8$.
    
    Several techniques have been proposed over the years to mitigate this problem by increasing either the unambiguous velocity region or the unambiguous range region.
    A first approach uses carrier frequency variation and transmits pulses with different carriers.
    The velocity ambiguous region is increased by exploiting phase differences between pairs of reflected pulses \cite{7485194}.
    However, it is not clear how to compute the phase differences in the presence of more than one target.
    This technique suffers from additional issues, including radar cross section (RCS) variation under different carriers and large frequency excursion requirement \cite{ludloff1981doppler}.
    Therefore, methods based on pulse repetition frequency (PRF) variation are generally preferred \cite{599236}, where $\mathrm{PRF} = 1 / {T_r}$.
    
    Two main PRF variation based techniques are staggered PRFs and multiple PRFs (MPRF).
    The use of staggered PRFs has been essentially proposed to raise the first blind speed $V_{\max}$ significantly without degrading unambiguous range \cite{richards2005fundamentals}.
    Pulse-to-pulse stagger varies the PRF from one pulse to the next pulse, achieving increased Doppler coverage \cite{4101740, 407142}.
    The main disadvantage of this approach is that the data corresponds to a non-uniformly sampled sequence, making it more difficult to apply coherent Doppler filtering \cite{richards2005fundamentals}. 
    In addition, clutter cancellation also becomes more challenging and the sensitivity to noise increases \cite{599236, 7485194}.
    
    The MPRF approach transmits several pulse trains, each with a different PRF.
    Ambiguity resolution is typically achieved by searching for coincidence between unfolded Doppler or delay estimations for each PRF.
    A popular approach, adopted in \cite{4104095}, relies on the Chinese Remainder Theorem \cite{skolnik2008radar} and uses two PRFs, such that the numerator and denominator of the ratios between these are prime numbers.
    The ambiguous ranges are computed for each train and congruence between these are found by exhaustive search.
    However, in this approach, a small range error on a single PRF can cause a large error in the resolved range with no indication that this has happened \cite{270476}.
    
    Trunk et. al \cite{270476} propose a clustering algorithm which implements the search of a matching interval by computing average distances to cluster centers.
    This technique still requires exhaustive search of clusters and does not process the samples jointly, decreasing signal to noise ratio (SNR).
    An alternative method using a maximum likelihood criterion, which avoids the use of matching intervals, has been proposed for Doppler ambiguity resolution \cite{599236}.
    This algorithm, which relies on the choice of particular values for the PRFs, first estimates the folded or reduced frequency and then uses it to estimate the ambiguity order.
    However, it has been demonstrated that the ambiguity order estimation is very sensitive to the folded frequency estimation preformed initially \cite{407142}.
    
    In this paper, we adopt a random pulse phase coding (RPPC) approach to increase the unambiguous range region, while preserving the unambiguous Doppler region  using a single PRF.
    RPPC has been used in polarimetric weather radars, which exploit the inherent random phase between pulses of the popular magnetron transmitters \cite{6032738}. In this context, RPPC mitigates out-of-trip echoes \cite{6032738}.
    In our approach, a random phase is introduced from pulse to pulse, and we then jointly process the received signals from all pulses to resolve range ambiguity.
    
    Our work has three main contributions.
    First, theoretical analysis is performed on  unambiguous target recovery conditions  in the noiseless case.
    For a given ambiguous delay region $[0, Q T _ r) $ with an integer  $Q > 1$, it is proved that range ambiguity can be resolved with sparse recovery methods if the number of targets in each ambiguous range resolution bin is less than $ (P-Q+2)/2$, where $P$ is the number of transmit pulses.
    Second, compared with MPRF method, our approach improves  SNR by jointly processing the samples from the overall received signal, rather than matching the estimated parameters from each pulse train processed separately.
    Therefore, our approach achieves improved delay and Doppler estimation over the MPRF methods.
    Finally, we use the matrix version of orthogonal matching pursuit (OMP) \cite{2013arXiv1311.2448W,4385788} for unambiguous delay-Doppler recovery,  which does not involve exhaustive search.
    From a practical point of view, our technique does not require the use of different PRFs, simplifying hardware implementation.
    
    In addition, our recovery algorithm can be directly applied to compressed samples, obtained using the sub-Nyquist method proposed in \cite{6733283, 6850183, Cohen2018}.
    This scheme exploits the sparse nature of radar target scenes to overcome the sampling rate bottleneck, breaking the link between radar signal bandwidth and sampling rate.
    In \cite{6733283, 6850183, Cohen2018}, the Fourier coefficients of the received signal are obtained from low-rate point wise samples taken after analog pre-filtering.
    The delay-Doppler map may then be recovered using compressed sensing (CS) algorithms \cite{eldar2012compressed, eldar2015sampling}.
    Our CS based unambiguous delay-Doppler recovery method can be applied to these compressed samples, without requiring any modification. 
    Given the number of samples $K$ within each PRI, an upper bound on the number of targets for unambiguous target recovery is given by $ \min \{ (K+1)/2, (P-Q+2)/2 \} $.
    
    We compare our approach to the popular MPRF method of \cite{270476}, which has been shown to outperform the matching interval scheme based on the Chinese Remainder Theorem.
    We demonstrate that our algorithm outperforms  MPRF  both in Nyquist and sub-Nyquist regimes.
    
    The rest of the paper is organized as follows.
    In Section~\ref{section2}, we present the random phase coded pulse radar model with range ambiguity, introduce the corresponding sampling methods, and establish the range-Doppler recovery model.
    Section~\ref{section4} introduces our unambiguous delay-Doppler recovery algorithm based on OMP.
    Section~\ref{section-condition} presents a  theoretical analysis of the  unambiguous delay-Doppler recovery in the noiseless case.
    Simulation results are provided in Section~\ref{section5}. 
    We conclude in Section~\ref{section6}. 
    
    \emph{Notation:} 
    For a vector $\x$, a matrix $\X$, and  positive integers $i$ and $j$, the $i$-th element of $\x$ is denoted by $\x_i$, the $j$-th column of $\X$ is denoted by $\X_{ j }$, and the $(i,j)$-th element of $\X$ is written as $\X_{i,j}$.
    Here, the element index begins with zero. 
    For instance, the first element of $\x$ is $\x_0$, and the first column of $\X$ is $\X_0$. 
    Given integers $N,m,n$, $W_N^{mn}$ represents $e ^ {-j 2 \pi m n / N}$.
    In this paper, $(\cdot) ^ \mathrm{H}$, $(\cdot) ^ \mathrm{T}$, $(\cdot) ^ c$ and $(\cdot) ^ \mathrm{-1}$ are the  Hermitian transpose, transpose, conjugate and inverse, respectively.

    \section{Problem formulation} \label{section2}
    
    In this section, we  first present the signal model of random phase coded pulse radar.
    Then, we introduce the sampling schemes for radar echoes, in both Nyquist and sub-Nyquist regimes.
    Finally, we formulate the sparse matrix recovery problem for range-Doppler recovery, which will be used to derive the recovery method and recovery conditions in the following sections.
    
    \subsection{Signal model} \label{section2-1}
    
    In the signal model, a pulse-Doppler radar transceiver transmits a phase-coded pulse train consisting of $P$ equally spaced pulses.    
    For $0 \leq t \leq P T_r$, this pulse train is given by 
    \begin{equation} \label{eq-ch2-tranmit-signal}
        s(t) = \sum_{p=0}^{P-1} h(t - p T_r) e ^ {j  \phi[p]  } e ^ {j 2 \pi f _c t},
    \end{equation}
    where $h(t)$ is the time-limited baseband waveform taking nonzero values in the interval $ [ 0 , T_h  ) $ (with $T_h$ being the pulse width), the pulse-to-pulse delay $T_r$ is the PRI, and $f_c$ is the carrier frequency.
    We use $\phi[p]$ to represent the phase shift of the $p$-th pulse, for $p = 0,\ldots,P-1$. As opposed to the traditional pulse Doppler radar, where the phase codes are identical, here in the random phase coded pulse radar, $\phi[p]$ is randomly distributed in the interval $[0,2 \pi)$ and varies from pulse to pulse.
    The entire span of the	signal in \eqref{eq-ch2-tranmit-signal} is called the coherent processing interval (CPI).
    We also assume that $h(t)$ is band-limited, and $B_h$ is referred to as the bandwidth of $h(t)$.
    
    \color{black}
    
    \Copy{KeySignal1}{Consider that the radar illuminates a point target moving with radial velocity $V$, whose distance to the radar  is given by $R(t) = R(0) - Vt$.
    The echo signal from the target is} \cite{doi:10.2200/S00170ED1V01Y200902SPR008}
    \begin{equation} \label{eq-rt}
            r(t) = \sum_{p=0}^{P-1} \alpha h(t - p T_r - \tau(t)) e ^ {j  \phi[p]  } e ^ {j 2 \pi f _c (t - \tau(t))},
    \end{equation}
   \Copy{KeySignal2}{
    where  $ \tau(t)$ is the round-trip time delay and $\alpha$ is a complex amplitude factor accounting for the antenna gain, the two-way path loss and the target's RCS.
    From \cite{doi:10.2200/S00170ED1V01Y200902SPR008}, the time delay can be approximately given by}
    \begin{equation} \label{eq-taut}
        \tau(t) \approx \tau - \frac{2 V}{c} t,
    \end{equation}
\Copy{KeySignal3}{
    if $V \ll c$, where $\tau = 2 R (0) / c$.
    Combining \eqref{eq-rt} and \eqref{eq-taut},}
    \begin{equation} \label{eq-rt-2}
        \begin{aligned}
           & r(t) = \\ & \sum_{p=0}^{P-1} \alpha h \left( (1+2V/c) t - p T_r - \tau \right) e ^ {j  \phi[p]  } e ^ {j 2 \pi (f _c t - \nu t - f_c \tau ) } ,
        \end{aligned}
    \end{equation}
\Copy{KeySignal4}{
     where $\nu = -2 V f_c / c$ is defined as the Doppler frequency.
     In \eqref{eq-rt-2}, the velocity stretches or compresses the envelop of the pulse train by the factor $1+2V/c$  \cite{doi:10.2200/S00170ED1V01Y200902SPR008}.
     When $V \ll c$, this effect is negligible, and the signal after down-conversion is}
     \begin{equation} \label{eq-rd}
         r(t) e ^ {-j 2 \pi f_c t} = \sum_{p=0}^{P-1} \alpha h \left( t - p T_r - \tau \right) e ^ {j  \phi[p]  } e ^ {-j 2 \pi \nu t} , 
     \end{equation} 
 \Copy{KeySignal5}{
    where $\alpha$ incorporates the factor $ e ^ {-j 2 \pi f_c \tau} $.}
    
    \color{black}
    
    Next we consider a target scene with $L$ point targets   located within the radar coverage region.
    The $l$-th target is defined by three parameters: a time delay $\tilde{\tau}_l = 2R_l / c$, where $R_l$ is the  distance from the radar to the target at $t = 0$; a Doppler frequency $\nu _ l = 2 V_l / \lambda $, where $V_l$ is the  radial velocity of the target; and a complex amplitude factor of the echo signal  $\alpha _ l$.
    \color{black}
    \Copy{KeyRCS}{The targets are assumed to have non-fluctuating RCSs, or have slowly-fluctuating RCSs, e.g. satisfying the Swerling-1 model \cite{1057561,skolnik2008radar}, and hence $\alpha _ l$ is constant during the CPI.}
    \color{black}
    The targets are defined in the radar radial coordinate system and the Doppler frequencies are assumed to lie in the unambiguous frequency region, that is $ \nu_l \in [0, 1/T_r )$, for $l = 0,\ldots,L-1$.
    As opposed to the common assumption in traditional radars, the time delays $\tilde{\tau}_l$ are not assumed to lie in the unambiguous region, namely less than $T_r$, but may exceed $T_r$, and  range ambiguity  occurs for a conventional pulse Doppler radar.
    For convenience, we decompose $\tilde{\tau}_l$ into its integer part (the ambiguity order) $q_l $ and the fractional part (the folded or reduced delay) $\tau_l$ as
    \begin{equation} \label{eq-ch2-delay}
        \tilde{\tau}_l = \tau_l + q_l T_r,
    \end{equation}
    where $ q_l \geq 0$ is an integer and $0 \leq \tau_l < T_r$. \color{black}
   \Copy{KeyAmbiguity}{Range ambiguity may occur in radars that have wide observation range and transmit pulses with high PRF for considerations such as: (a) avoiding Doppler ambiguity for high frequency radars; (b) increasing the integrated power for low peak power radars \cite{5089550}; (c) increasing the data rate in joint radar-communication systems \cite{9093221}. } \color{black}
    
    %
    
    
    \Copy{KeySignal6}{From \eqref{eq-rd}, the received signal after down-converting is written as}
    \begin{equation} \label{eq-ch2-receive-signal}
             y(t) =  \sum_{l=0} ^ {L-1} \sum_{p=0}^{P-1}   \alpha_l h \left(t - \tau_l - (p+q_l) T_r \right) e ^ {-j 2 \pi \nu_l t} e ^ {j  \phi[p]  }  + u(t) , 
    \end{equation}
    \Copy{KeySignal7}{
    where $u(t)$ is additive white Gaussian noise (AWGN) with variance $\sigma^2$.	
    Under the reasonable assumption {\color{black} $ \max_l |\nu_l| \ll 1/T_h $}, $y(t)$ can be approximated as \cite{doi:10.2200/S00170ED1V01Y200902SPR008}}
    \begin{equation} \label{eq-ch2-received-signal-approx}
        \begin{aligned}
            & y(t) = \\ &
            \sum_{l=0} ^ {L-1} \sum_{p=0}^{P-1}   \tilde { \alpha }_l h \left(t - \tau  _l - (p+q_l) T_r \right) e ^ {-j 2 \pi \nu_l (  p  + q_l ) T_r } e ^ {j  \phi [p]  }  + u(t) ,
        \end{aligned}
    \end{equation} 
   \Copy{KeySignal8}{
    for $0 \leq t \leq P T_r $, where {\color{black} $  \tilde { \alpha }_l  = { \alpha}_l e ^ {-j 2 \pi \nu_l \tau_l}$}.}

    \begin{figure}
        \centering
        \includegraphics[width=0.9\linewidth]{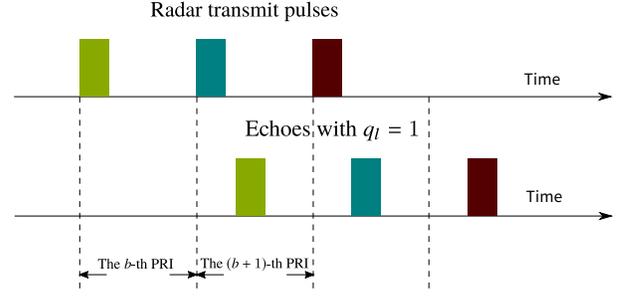}
        \caption{The pulse transmitted in the $b$-th PRI is received in the $(b+1)$-th PRI.}
        \label{fig:pulse}
    \end{figure}
    
    For convenience, we rewrite the overall received signal in \eqref{eq-ch2-received-signal-approx} with respect to each PRI.
    Note that in traditional pulse Doppler settings, namely under the assumption that $ 0 \leq \tilde{ \tau }_l < T_r $, the $p$-th pulse reflected from the targets is received in the $p$-th PRI.
    Here, the $p$-th pulse reflected from the $l$-th target is received in the $(p + q_l)$-th PRI.
    Figure \ref{fig:pulse} illustrates this phenomenon for $q_l = 1$, in which the $b$-th pulse is received in the $(b+1)$-th PRI.
    In other words, the $(b-q_l)$-th pulse reflected from the $l$-th target is received in the $b$-th PRI.
    Substituting $p = b - q_l$, we can  rewrite \eqref{eq-ch2-received-signal-approx} as 
    \begin{equation}  \label{eq-ch2-received-signal2}
        \begin{aligned}
            y(t) &= \sum_{l=0} ^ {L-1} \sum_{b=q_l} ^ {P+q_l-1} \tilde{\alpha}_l h(t - \tau_l - b T_r) e ^ {-j 2 \pi \nu _ l b T_r} e ^ {j  \phi[b-q_l]  } + u(t) \\
            &=   \sum_{b=0} ^ {P-1} \sum_{l=0} ^ {L-1}  \tilde{\alpha}_l h(t - \tau_l - b T_r) e ^ {-j 2 \pi \nu _ l b T_r} z[b-q_l] + u(t),
        \end{aligned}
    \end{equation}
    for $0 \leq t \leq P T_r $, where the sequence $\{z[p] \}$ is defined as
    \begin{equation}
        z[p] = \left \{  
        \begin{array}{cl }
            e ^ {j \phi [p]}, & \mathrm{for} \ p = 0,\ldots,P-1, \\ 
            0 , & \mathrm{for} \ p < 0 \ \mathrm{or} \  p > P-1    .
        \end{array} 
        \right.
    \end{equation}

    From \eqref{eq-ch2-received-signal2}, the received signal in the $b$-th PRI is expressed as 
    \begin{equation}  \label{eq-ch2-received-burst}
        y_b(t) = \sum_{l=0}^{L-1} \tilde{\alpha}_l h(t - \tau_l - b T_r) e ^ {-j 2 \pi \nu_l b T_r } z[b-q_l] + u_b(t),
    \end{equation}
    for $b = 0,1,\dots, P-1$ and $b T_r \leq t < (b+1) T_r$. 

    Given the received signal $y_b(t)$, $b = 0, \ldots,P-1$, our goal is to recover the range and velocity of targets, namely the time delays $\{ \tilde{ \tau }_l \} $ and Dopplers $\{ \nu_l \}$,  $l = 0,\ldots, L-1$. 
    To recover these parameters, we  first sample the signal, as presented in the next subsection.
    
    

    \subsection{Sub-Nyquist sampling} \label{section2-2}
    To reduce the sampling rates, we apply sub-Nyquist sampling in fast time, namely sampling the signal in each PRI with sampling rate lower than the bandwidth $B_h$.
    Generally,  aliasing of frequency bands will occur if the sampling rate is below the signal bandwidth.
    Nevertheless, the proposed techniques for sub-Nyquist radar \cite{6733283, 6850183, Cohen2018} can obtain the necessary frequency information to recover target parameters without aliasing, by appropriate analog anti-aliasing filtering before sub-Nyquist sampling. 
    
    To see this, we compute the Fourier series representation of the aligned received signal in the $b$-th  PRI $y_b(t + b T_r)$ with respect to the period $[0, T_r)$.
    This  results in  \cite{6733283, 6850183}
    \begin{equation}  \label{eq-ch3-received-fourier-burst}
        \begin{aligned}
            & Y_b[m] = \\  
            &  \frac{1}{T_r} H\left(\frac{2 \pi m}{T_r}\right) \sum_{l=0} ^ {L-1} \tilde{\alpha}_l  e ^ {-j 2 \pi \nu_l b T_r} e ^ {-j \frac{2 \pi}{T_r} m \tau_l }   z[b - q_l] + U_b[m], \\
        \end{aligned}
    \end{equation}   
    for $m = 0,\ldots,N-1$, where  $H(\cdot)$ denotes the Fourier transform of $h(t)$, $N = \lfloor B_h T_r \rfloor$ is the number of Fourier samples, and $\{ U_b[m] \}$ are the Fourier  coefficients of $u_b(t + b T_r)$.
    Here, $\lfloor \cdot \rfloor$ represents the floor function.
    For convenience, let
    \begin{equation} \label{eq-ch3-received-fourier-burst-normalize-0}
        \begin{aligned}
            \tilde{Y}_b[m] &= \frac{T_r Y_b[m]}{  H(2 \pi m / T_r) } \\
            & = \sum_{l=0} ^ {L-1} \tilde{\alpha}_l  e ^ {-j 2 \pi \nu_l b T_r} e ^ {-j \frac{2 \pi}{T_r} m \tau_l} z [b - q_l] + \tilde{U}_b[m]
        \end{aligned}
    \end{equation}
    be the normalized Fourier  coefficients of the $b$-th PRI, where $ \tilde{U}_b[m] = T_r U_b[m] / H( 2 \pi m / T_r)  $. 
    
    From \eqref{eq-ch3-received-fourier-burst-normalize-0}, the target parameters $\{ \tau_l, \nu_l, q_l \}$ are contained in the normalized Fourier  coefficients $\tilde{ Y }_b [ k ]$.
    To recover the target parameters, a sub-Nyquist radar obtains the Fourier  coefficients from low rate samples of the received signal in each PRI.
    In this paper, we consider  Xampling \cite{mishali_xampling_2011, mishali_xampling_2011-1} based sub-Nyquist radar systems.
    For each PRI, Xampling allows one to generate an arbitrary subset 
    \begin{displaymath}
        \kappa = \left\{m_0, \ldots, m_{K-1} \right\} \subset \left\{ 0,\ldots,N-1 \right\}, 
    \end{displaymath}
    comprised of $K = | \kappa |$ frequency components, and obtain the corresponding Fourier coefficients $Y_b[m_k]$, for $k = 0,\ldots,K-1$, from $K$  point-wise samples of the
    received signal $y_b(t + b T_r)$ after appropriate analog pre-processing.
    The procedure of Xampling is shown in Fig. \ref{fig:sample}, in which the received signal is split into $K = | \kappa |$ channels.
    In the $k$-th channel, Xampling obtains $Y_b[m_k]$ by first  mixing $y_b(t + b T_r)$ with the harmonic signal $e ^ {-j (2 \pi / T_r) m_k t }$ and then integrating over the aligned receive window $[0, T_r)$.
    
   \color{black}
   \Copy{KeyXampling}
    {By applying Xampling at radar receivers, we can reduce the sample rate without affecting the range resolution if radar targets can be sparsely represented. Thus, the cost and complexity of the analog to digital converter (ADC) at radar receivers may be reduced, especially for wide-band radars.}
    \color{black}
    
    \color{black}
    \Copy{KeyNyquist1}
    {After sub-Nyquist sampling, the problem is to recover the targets' delays $\{ \hat{\tau}_l \}$ and Dopplers $\{ \nu_l \}$, from the compressed normalized Fourier series} 
    \begin{equation} \label{eq-ch3-received-fourier-burst-normalize}
        \tilde{ Y }_b[m_k] =     \sum_{l=0} ^ {L-1} \tilde{\alpha}_l  e ^ {-j 2 \pi \nu_l b T_r} e ^ {-j \frac{2 \pi}{T_r} m _k \tau_l} z [b - q_l] + \tilde{U}_b[m_k],
    \end{equation}
\Copy{KeyNyquist2}{
    for $b = 0,\ldots,P-1$ and $k = 0,\ldots,K-1$.
    We note that these Fourier series can also be obtained by conventional Nyquist sampling, if we implement the integration in Fig. \ref{fig:sample} by performing a discrete Fourier transform to the samples.
    Therefore, the signal model in \eqref{eq-ch3-received-fourier-burst-normalize} holds in both Nyquist and sub-Nyquist regimes.
    Unlike sub-Nyquist sampling, in  Nyquist sampling all the Fourier coefficient are obtained, i.e. $K=N$ and $m_k = k$, for $k = 0,\ldots,N-1$.}
    \color{black}
    
    \begin{figure}
        \centering
        \includegraphics[width=0.9\linewidth]{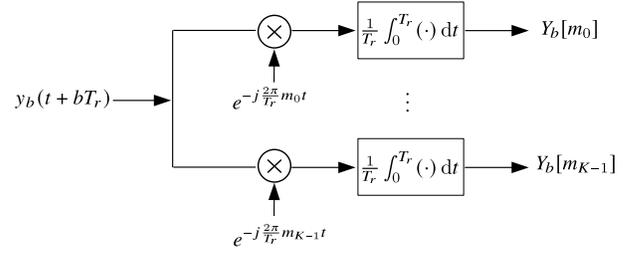}
        \caption{In Xampling \cite{Cohen2018}, the Fourier coefficients is directly sampled after analog pre-processing in each channel.}
        \label{fig:sample}
    \end{figure}

    \subsection{Matrix formulation} \label{section2-3}
    
    In this subsection, we recast \eqref{eq-ch3-received-fourier-burst-normalize} in  matrix form. 
    To that aim, we assume that the delays and Dopplers of the $L$ targets lie on the center of delay resolution bins and Doppler resolution bins, respectively.
    As in traditional pulse Doppler radar, the size of a delay resolution bin is $T_r / N$, while that of a Doppler resolution bin is $1/(PT_r)$.
    Then the delays and Doppler can be represented by $\tau_l = n_l T_r / N$ and $ \nu_l=  p_l / (P T_r)$, where $n_l$ and $p_l$ are integers in the intervals $[0,N-1]$ and $[0,P-1]$, respectively, for $ l = 0, \ldots,L-1$.
    Under this assumption, \eqref{eq-ch3-received-fourier-burst-normalize} becomes
    \begin{equation}  \label{eq-ch3-received-normalized-fourier-burst-2}
        \tilde{Y}_b[m_k] =  \sum_{l=0} ^ {L-1}  \tilde{\alpha}_l  W _ P ^{ b p _l } W_N ^{m_k n_l} z [b - q_l] + \tilde{U}_b[m_k],
    \end{equation}
    for $b = 0,\ldots, P-1 $ and $k = 0,\ldots, K-1$.  
    
    Define the ambiguity factor
    \begin{eqnarray} \label{eq-defineQ}
        Q = \max (q_0, \ldots, q_{L-1}) + 1.
    \end{eqnarray}
    When $Q = 1$, there is no range ambiguity.
    In our model, range ambiguity is considered, i.e. $Q > 1$.
    In \eqref{eq-ch3-received-normalized-fourier-burst-2}, the target parameters $\{\tilde{ \alpha} _l , n_l, q_l, p_l \} $ can be characterized by a matrix $\tilde{\X} \in \mathbb{C}^{N \times PQ}$, which is defined as
    \begin{equation}
        \tilde{\X}_{n, c } = \left\{  \begin{array}{cl}
            \tilde{ \alpha_l }, & \ \mathrm{if} \ n = n_l \ \mathrm{and} \  c = Pq_l + p_l,   \\
            0, & \ \mathrm{otherwise}.
        \end{array} \right.
    \end{equation}
    In other words, the matrix $\tilde{\X}$ is an $N \times PQ$ matrix which contains the value $\tilde{ \alpha} _l $ at the corresponding $L$ indexes $(n_l, P q _ l + p_l)$, for $l = 0,\ldots,L-1$, while the rest of the elements in $\tilde{\X}$ are all zeros.

    We may now reformulate \eqref{eq-ch3-received-normalized-fourier-burst-2} into a matrix observation model
    \begin{equation} \label{eq-ch3-received-fourier-matrix-2}
        \Y = \A \tilde{\X} \B  ^ \mathrm{T} + \U,
    \end{equation}
    where the $(k,b)$-th entry of $\Y \in \mathbb{C}^{K \times P}$ is given by $ \tilde{Y}_b[m_k]$, denoting the $k$-th Fourier coefficients of the radar signal received in the $b$-th PRI, for $b = 0,\ldots,P-1$ and $ k = 0,\ldots,K-1$; the partial Fourier matrix $\A \in \mathbb{C}^{K \times N}$ has the $(k,n)$-th entry given by $  W_N^{n m _ k } $, representing the fast-time frequency response from the $n$-th range resolution bin at the frequency point $m_k / T_r$, for $n = 0,\ldots, N-1 $ and $k = 0,\ldots,K-1$; and $\U \in \mathbb{C}^{K \times P}$ is the additive noise whose $(k,b)$-th entry is given by $ \tilde{ U }_b[m_k]$.
    
    In \eqref{eq-ch3-received-fourier-matrix-2}, the  matrix $\B \in \mathbb{C}^{ P \times PQ }$ consists of $Q$ blocks. Particularly, $\B$ is represented by
    \begin{equation}
        \B = \left[  \B ^ {(0)}, \B ^{(1)} , \ldots, \B^{(Q-1)}  \right] ,
    \end{equation}
    where $\B^{(q)} \in \mathbb{C}^{P \times P}$, and the $(b,p)$-th entry of $\B^{(q)}$ is given by $W_{P} ^ {bp} z[b-q]$, for $p = 0,\ldots,  P-1$ and $b = 0 ,\ldots, P-1$. 
    Here, each block $\B^{(q)}$ represents the slow-time response of the targets with ambiguity order $q$.
    
    From the matrix formation model \eqref{eq-ch3-received-fourier-matrix-2}, $\tilde{\X}$ should be a solution of the following equation 
    \begin{equation}  \label{eq-inverse_problem}
        \A \X \B ^ \mathrm{T} = \Y , 
    \end{equation}
    in the noiseless case.  
    The problem  is to recover the sparse matrix $\tilde{\X}$ from the observation $\Y$ and  measurement matrix $\A$ and $\B$, by finding the solution of \eqref{eq-inverse_problem}.
    For a Nyquist pulse-Doppler radar without range ambiguity, namely $K=N$ and $Q=1$,  $\A$ and $\B$ are full-rank square matrices, so the solution to \eqref{eq-inverse_problem} is unique.
    However, in our setting, namely $ K \leq N$ and $Q>1$, due to the  rank deficiency of  $\A$ and $\B$,  \eqref{eq-inverse_problem} is an under-determined equation and may not have a unique solution.	

    \color{black}
    \Copy{KeySparse}
     {Nevertheless, when $L \ll NPQ$, there are only a few nonzero elements in $\tilde{\X}$, which means that $\tilde{\X}$ is a sparse matrix. 
   This sparsity of radar targets motivates the use of CS algorithms to solve the under-determined radar observation model.   
   In recent years, CS algorithms have been applied to many fields of radars, such as synthetic aperture radar  imaging \cite{6850193,9170786}, space-time adaptive processing \cite{8738977} and randomized stepped frequency radars  \cite{8494717,9354050} and  exhibit enhanced target reconstruction quality compared to a matched filter on real radar data \cite{8938756,9354050,9170786}.
   In addition, various  low complexity methods \cite{8863507,8938756} are proposed for the real-time  implementation of CS algorithms on radars.
   In our problem, CS algorithms may be applicable for ground-to-air radars, where the sparsity of targets holds and the computation complexity of CS algorithms is affordable.
   We follow the concepts of CS and use sparse matrix recovery  to recover $\tilde{\X}$ from the received signal, as  discussed in the following section.} 
  \color{black} 
    
    %

    \section{Delay-Doppler Recovery Methods} \label{section4}
    
    \color{black}
    \Copy{KeyCS1}{
        To recover  $\tilde{\X}$ from \eqref{eq-ch3-received-fourier-matrix-2}, we consider the $\ell_0$ ``norm'' minimization problem}  
    \begin{equation} \label{eq-ch3-recovery-problem}
        \min \ \| \X \|_0, \ \ \mathrm{s.t.} \ \  \A \X \B  ^ \mathrm{T} =  \Y ,
    \end{equation}
     \Copy{KeyCS2}{
    under the assumption that $\X$ is a sparse matrix. 
    Here, $\| \cdot \|_0$ represents $\ell_0$ ``norm'', which is defined as the number of nonzero elements of a vector or a matrix.
    The $\ell_0$ ``norm'' is a non-convex function and the sparse matrix recovery problem \eqref{eq-ch3-recovery-problem} is generally NP hard.
    Therefore, solving \eqref{eq-ch3-recovery-problem} is computationally intractable in practical problems.
    A more practical way is to compute a sub-optimal solution with heuristic greedy methods such as OMP \cite{4385788} and iterative hard thresholding  \cite{blumensath_iterative_2008,BLUMENSATH2009265}.
    The problem in \eqref{eq-ch3-recovery-problem}  can also be solved by relaxing the $\ell_0$ ``norm'' minimization into the convex $\ell_1$ norm minimization, which was shown to be tight under specific conditions \cite{4472240}.  
    
    Considering the computation complexity, we use the matrix version of OMP  to solve \eqref{eq-ch3-recovery-problem}.
     Matrix OMP recovers $L$ non-zero elements in $X$ with $L$ iterations.
    In the $t$-th iteration, the location of a new non-zero element in $\X$ is first estimated by a matched filter, then the values of all the non-zero elements is updated by least squares estimation, and finally the signal residual is updated by subtracting the signals of all the non-zero elements.
    The detailed procedure of matrix OMP are omitted here and can be found in \cite{2013arXiv1311.2448W, 8361480}.
    Once $\bm{X}$ is recovered, let $\Lambda_{l,1}$ be the row index and $\Lambda_{l,2}$ be the column index of the $l$-th non-zero element in $\X$, respectively.}
    \Copy{KeySolution}{Then the delay ambiguity orders, folded delays and Dopplers are estimated as
      \begin{displaymath}
            \hat{q}_l =  \left \lfloor \frac{  \Lambda_{l,2} } {P} \right \rfloor, \ \hat{ \tau}_l =  \frac{T_r}{N} \Lambda_{l,1} , \ \hat{ \nu } _ l = \frac{\Lambda_{l,2}   - \hat{q}_l P}{P T_r},
      \end{displaymath}
        where $\lfloor (\cdot) \rfloor$ is the floor function.
      We note that here the range ambiguity is not explicitly resolved, but is indirectly resolved by solving  \eqref{eq-inverse_problem}.  
     Similarly, other CS recovery algorithms, such as FISTA \cite{doi:10.1137/080716542,5651522}, can be extended to our setting, namely to solve \eqref{eq-ch3-recovery-problem}.}

   \Copy{KeyComplexity}
    {The computational complexity of OMP is higher than traditional radar processing techniques like matched filter, which is performed with $\mathcal{O}( NKP + NP^2 Q )$ computations. 
    In the $t$-th  iteration of matrix OMP, the matched filter needs $\mathcal{O}( NKP + NP^2 Q )$ computations, the least squares estimation needs $\mathcal{O}( t^3 + t^2(K+P) )$ computations, and the complexity for residual update is $\mathcal{O}( tKP )$.
    The  complexity of matrix OMP to recover $L$ targets is then
    \begin{displaymath}
        \begin{aligned}
            &	\mathcal{O} \big( \sum_{t=1}^{L} t^3 + t^2(K+P) + tKP +  NKP + NP^2 Q \big) \\
            & = \mathcal{O} ( L^4 + L^3 (K+P) + L^2 KP + L NKP +L  NP^2 Q ).
        \end{aligned}
    \end{displaymath}
    If $L < N$ and $L <P$, the computation complexity becomes $\mathcal{O}(L NKP + L NP^2 Q )$, which is $L$ times the complexity of a matched filter.}

   \Copy{KeyOffGrid}{In our derivations, it is assumed that the delays and Dopplers lie at the center of delay resolution bins and Doppler resolution bins, respectively.
    However,  real radar parameters are defined in a continuous domain and can be ``off the grid'' \cite{6576276}, namely do not lie at the center of resolution bins.
    In this case, the signal $\bm{Y}$ may not be sparsely represented by \eqref{eq-inverse_problem}, leading to reconstruction error for sparse recovery methods.
    To overcome this problem, a simple strategy is to reduce the size of the grids.
    In particular, given an  over-discretization factor $ \gamma \geq 1$, we can reduce the size of range grids and Doppler grids to $ T_r / (\gamma N)  $ and $ 1 / (\gamma P T_r)$, respectively.
    When $\gamma$ is large enough, the continuous parameters approximately lie on the grid and the reconstruction error due to the off-grid effect can be eliminated.
    The main problem of this strategy is the increase of computation complexity, especially when $\gamma$ is large, as the number of columns in $\bm{A}$ and $\bm{B}$ increases.
    Alternatively, several sparse recovery schemes are newly proposed that do not involve  discretization  and directly recover the parameters in a continuous domain, such as atomic norm minimization \cite{6576276,9016105} and alternating descent conditional gradient \cite{boyd_alternating_2017}.}

     \color{black}

    \section{Delay-Doppler Recovery Conditions} \label{section-condition}
    
    In this section, we show that the range and Doppler parameters of radar targets can be unambiguously recovered under certain conditions in the noiseless case.
    Specifically, we  derive  conditions with respect to the number of targets, under which $\tilde{\X}$ can be unambiguously recovered by solving  \eqref{eq-ch3-recovery-problem} in the sub-Nyquist regime. 
    We begin with reformulating \eqref{eq-ch3-recovery-problem} in vector form, followed by some preliminaries on CS, and then derive the delay-Doppler recovery conditions.

    \subsection{Recovery condition in the sub-Nyquist regime}
    \label{sec-condition1}
    
    To derive the recovery conditions, we equivalently rewrite $\Y = \A \X \B ^ \mathrm{T}$ in vector form as
    \cite{2009arXiv0902.4587J}
    \begin{equation} \label{eq-ch3-received-fourier-matrix}
        \mathrm{vec}( \Y)  = ( \B   \otimes { \A } ) \mathrm{vec} (  \X ) ,
    \end{equation}
    where the operator $\mathrm{vec}(\cdot)$ produces a vector by stacking columns of a given matrix and $ \otimes $ represents the Kronecker product.
    Correspondingly, the $\ell_0$  minimization problem \eqref{eq-ch3-recovery-problem} becomes 
    \begin{equation} \label{eq-ch3-recovery-problem-2}
        \min \ \| \x \|_0, \ \ \mathrm{s.t.} \ \  \y = \T \x,
    \end{equation}
    by letting $\x=  \mathrm{vec} (\X)$, $\y =  \mathrm{vec}(\Y)$ and 
    \begin{equation} \label{eq-TBA}
        \T = \B  \otimes { \A } 
    \end{equation}
    in \eqref{eq-ch3-recovery-problem}.
    Let $\tilde{\x} = \mathrm{vec} ( \tilde{\X})$. 
    To unambiguously recover $\tilde{\X}$, we need $\tilde{\x}$ be the unique optimum of \eqref{eq-ch3-recovery-problem-2}.    
    
    CS theory provides conditions for  recovering $\tilde{\x}$ with \eqref{eq-ch3-recovery-problem-2} by investigating the spark property of the measurement matrix $\T$.
    The spark of $\T$ is defined as the size of the smallest linearly dependent subset of columns, i.e.
    \begin{equation} \label{eq-T}
        \mathrm{spark} (\T) = \min \left\{  \| \x \|_0: \ \T \x = \bm{0}, \ \x \neq \bm{0}  \right\}.
    \end{equation}
   \Copy{KeyCondition}{From the definition of spark, $\tilde{\x}$ is the unique optimum of \eqref{eq-ch3-recovery-problem-2} if   $ \mathrm{spark} ( \T ) > 2  \| \tilde{\x} \|_0$ \cite{donoho_optimally_2003, eldar_2015}.}
    From \cite{2009arXiv0902.4587J}, one has that for $\T$ of \eqref{eq-T},  
    \begin{equation}
        \mathrm{spark} ( \T ) = \mathrm{spark}(\B   \otimes \A  )
        = \min \{ \mathrm{spark}( \A ),\mathrm{spark}( \B ) \}.
    \end{equation}
    Since $\| \tilde{\x} \|_0 =  \| \tilde{\X} \|_0 =  L$, the unambiguous recovery condition becomes
    \begin{equation} \label{eq-spark_condition}
        \mathrm{spark}( \A ) > 2 L,\ \mathrm{spark}(\B) > 2L.
    \end{equation}
    
    For convenience, let $\beta_A = \mathrm{spark} (\A)$ and $\beta _ B= \mathrm{spark} (\B)$.  
    A naive bound of $\beta_A$ is given by 
    \begin{equation} \label{bound1}
        \beta_A \leq K+1,  
    \end{equation}
    because $\A$ is a $K \times N$ matrix with $K \leq N$ and any $K+1$ columns of $\A$ are linearly dependent.    
    We observe that the last block in $\B$ only has $P-Q+1$ non-zero rows, meaning that any $P-Q+2$ columns from the last block are linearly dependent. 
    As a result, 
    \begin{equation} \label{bound2}
        \beta_B \leq P - Q + 2.
    \end{equation}
    
    Since  $\A $ is a partial Fourier matrix generated by selecting $K$ rows from a $N$-dimensional Fourier matrix indexed with the subset $\kappa \subset \left\{0,\ldots,N-1 \right\} $,  $\beta_A$ depends on $\kappa$.
    From \cite{eldar2015sampling, alexeev_full_2012, achanta_spark_2017}, we can easily generate subsets $\kappa$ to ensure that $\A$ has full spark, i.e. $\beta_A = K+1$.
    We note that when the received signal is sampled at the Nyquist rate, i.e. $K=N$ and $\kappa = \left\{ 0,\ldots,N-1 \right\}$, $\A$ becomes a full Fourier matrix and the $N$ columns of $\A$ are linearly independent.
    In this case,  $\beta_A = N+1 = K+1$.
    
    Next, we  note that $\B$ is a random matrix since each element in $\B$ includes a random phase item.
    Thus $\beta_B$  is a random variable with respect to $\{z[p] \}$. 
    Under the assumption that the random phase item $ \phi[p] $ is generated from a uniform distribution over $[ 0, 2 \pi)$, the spark of $\B$, $\beta_B$, almost surely equals  $P-Q+2$, as indicated in the following theorem.
    \begin{mytheorem} \label{Theorem::A}
        Suppose that $ \phi[p] $ is independently and uniformly distributed in $[0, 2 \pi)$, for $ p = 0,\ldots,P-1 $.
        Then, with probability one, $\beta_B = P-Q+2$.
    \end{mytheorem}
    
    \begin{proof}
        See Appendix.
    \end{proof}

    Combining the  results on $\beta_A$ and $\beta_B$ with the recovery condition \eqref{eq-spark_condition}, we obtain the following theorem.
    In \eqref{eq-condition},  $K$ is the number of samples in each PRI,  $P$ is the number of transmit pulses and $Q$ is the ambiguity factor defined in \eqref{eq-defineQ}.
    \begin{mytheorem} \label{theorem1}
        Assume that 1) The subset $\kappa$ is properly designed so that $\A$ has full spark; 2) The phase terms $ \{ \phi[p] \} $ are independently and uniformly distributed in $[0, 2 \pi)$.
        Suppose that there exist $L$ targets with maximal ambiguity factor $Q$.
        In the noiseless setting, the range and Doppler parameters of these targets can be unambiguously recovered with probability one by solving \eqref{eq-ch3-recovery-problem} or \eqref{eq-ch3-recovery-problem-2} if and only if  
        \begin{equation}\label{eq-condition}
            L < \min \left\{ \frac{K+1}{2}, \frac{P-Q+2}{2} \right\}.
        \end{equation}	
    \end{mytheorem}
    \begin{proof}
        From the assumptions, one has $\beta_A = K+1$ and $\beta_B = P-Q+2$ with probability one.
        The recovery condition then becomes
        \begin{equation} \label{eq-condition-2}
            K+1 > 2L,\ P-Q+2 > 2L.
        \end{equation}
        The condition in \eqref{eq-condition} can be directly obtained from \eqref{eq-condition-2}.
    \end{proof}

    The randomness of the phase codes $\{z[k]\}$ is crucial for the derivation of Theorem~\ref{Theorem::A} and \ref{theorem1}.
    In a pulse-Doppler radar without phase coding in which $z[0]=\cdots=z[P-1]$,  it can be validated that
    \begin{equation}
        \B_1 - \B_2 = \B_{P+1} - \B_{P+2},
    \end{equation}
    which means there exist $4$  linearly dependent columns in $\B$. As a consequence,  $\beta_B \leq 4$ and the  number of targets is bounded by $L < 2$ from \eqref{eq-spark_condition}.
    
    \subsection{Recovery condition in the Nyquist regime}
    
    In the Nyquist regime, $\A$ is an invertible Fourier matrix. Therefore, \eqref{eq-ch3-recovery-problem} becomes
    \begin{equation} \label{eq-ch3-recovery-problem-3}
        \min \ \| \X \|_0, \ \ \mathrm{s.t.} \ \  \B \X ^ \mathrm{T}    = \Y ^ \mathrm{T} \A ^ c  .
    \end{equation}
    Let $\X ^ \mathrm{T} = [\bm x_0,  \ldots, \bm x _ {N-1}  ]$.
    The  problem in \eqref{eq-ch3-recovery-problem-3} can be split into multiple independent sub-problems:
    \begin{equation} \label{eq-ch3-recovery-problem-4}
        \min \ \| \bm x _n \|_0, \ \ \mathrm{s.t.} \ \  \B \bm x _ n    = [\Y ^ \mathrm{T} \A ^ c]_n  .
    \end{equation}
    for $n = 0,\ldots,N-1$.   
    The corresponding recovery condition for these sub-problems are
    \begin{equation} \label{eq-condition-4}
        \| \bm x _ n \|_0 < \beta_B  / 2 = \frac{P-Q+2}{2}, \ 0 \leq n \leq N-1.
    \end{equation}
    The conditions in the Nyquist regime only require that the number of targets within each reduced range resolution bin is bounded by $(P-Q+2)/2$, and is much looser than that in the sub-Nyquist regime, in which the total number of targets is bounded by $(P-Q+2)/2$.

    We  note that the upper bound on the  number of targets in \eqref{eq-condition-4} is reduced compared to a conventional pulse Doppler radar.
    For a Nyquist pulse Doppler radar without range ambiguity, the matrix  $\B$ becomes an  invertible Fourier matrix.
    Thus, the targets can be recovered by directly solving \eqref{eq-inverse_problem}  without using CS.
    Under this circumstance, the number of recoverable targets in each range resolution bin is $P$.
    In our setting,  range ambiguity leads to rank deficiency of $\B$.
    As a result, the observation equation is under-determined, and is solved by sparse matrix recovery methods, for which the upper bound of the number of targets is given in \eqref{eq-condition-4}.

    \section{Numerical experiments} \label{section5}
    
    In this section, we present some numerical experiments illustrating our proposed unambiguous range-Doppler recovery algorithm.
    We compare our method with the classical MPRF algorithm from \cite{270476} and examine the impact of sub-Nyquist sampling as well as range ambiguity order $Q$ on the detection performance.
    
    \subsection{Preliminaries}
    
     \color{black}
    \Copy{KeyPreliminaries}{
    We consider a pulse Doppler radar transmitting a pulse train composed of $P = 20$ pulses with PRI $T_r = 25 \mu \mathrm{sec}$ over a CPI of $500 \mu \mathrm{sec}$.
    The carrier frequency is $f_c = 10 \mathrm{GHz}$ and the propagation velocity is $f_c = 3 \times 10 ^ 8 \mathrm{m/s}$.
    Then we have $R_\mathrm{max} = 3.75\mathrm{km}$ and $V_\mathrm{max} = 300\mathrm{m/s}$. 
    The baseband waveform $h(t)$ is a linear frequency modulation pulse  with bandwidth $B_h = 20 \mathrm{MHz}$  and pulse width $T_h = 1 \mu \mathrm{sec} $. 
    Specifically, the expression of $h(t)$ is
    \begin{displaymath}
        h(t) = \left\{    
        \begin{array}{cc}
        e ^ {j \pi (B_h / T_h) t ^ 2},    & 0 \leq t \leq T_h ,  \\
          0 , & \mathrm{otherwise}. 
        \end{array}      
        \right.
    \end{displaymath}
}

   \color{black}
   
    To extend the maximal unambiguous range, we adopt range phase coding to each pulse, where the phase $ \phi[p]  $ is uniformly distributed in $[0, 2 \pi)$, for $ p = 0,\ldots,P-1$.
   
    The number of Nyquist rate samples in each PRI is  $N = T_r B_h= 500$.
    In the simulations, we investigate sub-Nyquist sampling by reducing the number of samples $K$ in each PRI.
    We randomly select $K < N$ frequency components and obtain the corresponding compressed Fourier coefficients by the Xampling scheme in Fig. \ref{fig:sample}.
    We find that the matrix $\A$ has full-rank with high probability if the frequency components are selected randomly.

    
    We consider $L$ targets with Doppler frequencies spread uniformly at random in the appropriate unambiguous region $ \left[0, 1 / T_r \right) $ and delays spread uniformly at random in the ambiguous region $ \left[0, Q T_r \right)$ for ambiguity factor $Q $.
    In the simulations, the echoes from all targets have unit amplitude, i.e. $| \alpha_l | = 1$, for $l=0,\ldots,  L-1$.
    
     \color{black}
    \Copy{KeySignal}
   { In the simulations, we produce the received signal $ y(t) $  with \eqref{eq-ch2-receive-signal}.
    The received signal is corrupted with AWGN $u(t)$ which has  variance $\sigma^2$ and is band-limited to $B_h$.    
    The total transmit SNR of the transmitted pulse train is 
    \begin{displaymath}
    	\mathrm{SNR} = \frac{ P \int_{0}^{T_h} |  h(t) | ^ 2 \, \mathrm{d} t } {\sigma^2}.
    \end{displaymath}
    Here, the inter-pulse random phase coding does not affect the  SNR after coherent integration.

    After $y(t)$ is produced, we  compute  $Y_b[m] $ and $H(2\pi m / T_r)$ by applying Fourier transform to $y_b(t)$ and $h(t)$, respectively.
    Then the elements of matrix  $\Y$ are computed via \eqref{eq-ch3-received-fourier-burst-normalize-0} and
    \eqref{eq-ch3-received-fourier-burst-normalize}, after which the delays and Dopplers of the targets are recovered by solving \eqref{eq-ch3-recovery-problem} via matrix OMP. }   \color{black}
    We use a hit-or-miss criterion as a performance metric. A ``hit'' is defined as a delay-Doppler estimate circumscribed by a rectangles around the true target position in the time-frequency plane.
    We use rectangles with axes equivalent to $\pm 1 $ times the delay and Doppler resolution bins, equal to $1 / B_h = 50 \mathrm{nsec}$ and $1 / (P T_r) = 2 \mathrm{KHz}$, respectively.

    
    \subsection{Comparison to MPRF scheme}
    
    We compare our approach to the popular MPRF method of \cite{270476} that has been shown to outperform the matching interval scheme based on the Chinese Remainder Theorem.
    In  MPRF, the pulse Doppler radar transmits two pulse trains with baseband signal $h(t)$.
    The first train is composed of $P_1 = 20$ pulses, with PRI $T_{r,1} = 25 \mu \mathrm{sec}$ over a CPI of $500 \mu \mathrm{sec}$.
    The second train is composed of $P_2 = 25$ pulses, with PRI $T_{r,2} = 20 \mu \mathrm{sec}$ over a CPI of $500\mathrm{\mu sec}$.
    Like the random phase coded pulses, the observation model in \eqref{eq-ch3-received-fourier-matrix-2} still holds for each pulse train in the MPRF scheme, where the measurement matrix $\B$ is constructed for $Q=1$ and all $z[p] = 1$.
    We use matrix OMP to recover the ambiguous delay-Doppler map from each pulse train.
    Once the targets' Doppler frequencies and ambiguous delays are recovered, we apply the clustering method in \cite{270476} to  estimate the unambiguous delays.
    The total transmit SNR of the two transmit pulse trains is 
    \begin{displaymath}
    	\mathrm{SNR} = \frac{ (P_1 + P_2) \int_{0}^{T_h} |  h(t) | ^ 2 \, \mathrm{d} t } {\sigma^2}.
    \end{displaymath}
    
    In this experiment, the number of targets is $L = 5$ and  ambiguity factor $Q = 4$.
    Here, we require the Dopplers and delays lie in the center of Doppler and range resolution bins, respectively.
    Performance of RPPC and MPRF is compared with the same range resolution, Doppler resolution and total transmit SNR  so that the comparison is fair.
     
     Figure~\ref{fig:p1} presents  the delay-Doppler recovery performance of both MPRF and RPPC with respect to transmit SNR. 
     The results are obtained in both Nyquist  and sub-Nyquist regimes.
     In  the sub-Nyquist regime, we randomly choose $K = 250$ and $K = 125$ Fourier coefficients in each PRI, leading to a  compression ratio of $50\%$ and $25\%$, respectively.
     As Fig. \ref{fig:p1} shows, the hit rate increases with the increase of the total transmit SNR, which is proportional to the SNR after matched filtering in fast and slow time.
     We observe that our RPPC approach outperforms the MPRF approach in both Nyquist and sub-Nyquist regimes, in terms of the hit rate under the same total transmit SNR.
     The explanation is that the RPPC approach jointly processes all the received samples, while the MPRF approach processes the received samples of the two pulse trains separately.
     Therefore, the RPPC approaches can obtain better SNR after fast and slow time matched filter with the same total transmit SNR.
     
     To achieve the same hit rate as  MPRF, our RPPC technique requires a lower total transmit SNR, leading to around $3$dB SNR gain.
     As a result, for a radar transmit system with fixed pulse width, peak power and PRF,  RPPC  needs a lower number of transmit pulses to achieve a commensurate performance with  MPRF, and thus reduces the cost of power and transmit time.
     
     The impact of sub-Nyquist sampling is also demonstrated in Fig. \ref{fig:p1}. It is observed that the recovery performance in the Nyquist regime is better than that in the sub-Nyquist regime, and the recovery performance in sub-Nyquist regime decreases as the number of samples decreases. This is because sub-Nyquist sampling leads to  loss of SNR. 
     
       \begin{figure}
     	\centering
     	\includegraphics[width=\linewidth]{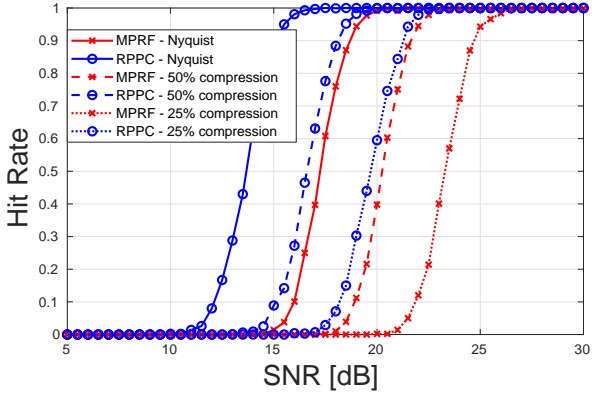}
     	\caption{Delay-Doppler recovery performance for MPRF and RPPC in  Nyquist  and sub-Nyquist regimes.}
     	\label{fig:p1}
     \end{figure}


    \subsection{Performance in the off-grid case}
    
     \color{black}
    \Copy{KeySimulationOffGrid}{The last experiment was conducted in the on-grid case, namely the delays and Dopplers lie in the center of the resolution bins.
    In this experiment, we consider a more realistic scene where the delays and Dopplers do not necessarily lie in the center of resolution bins, and examine the performance of matrix OMP in the off-grid case.
    In particular, there are $L = 5$ point targets whose Doppler frequencies and delays can be arbitrary values in the region $ \left[0, 1 / T_r \right) $ and $ \left[0, Q T_r \right)$ for $Q = 4$, respectively.
    
   In the off-grid case, we directly produce the Fourier coefficients in \eqref{eq-ch3-received-fourier-burst} for convenience. The hit rate of matrix OMP is computed for over-discretization factor $ \gamma = 1,2,4,16$.
   When $\gamma = 1$, there is no over-discretization, i.e. the range and Doppler grid size are equal to the range and Doppler resolution, respectively.
   The result for Nyquist and sub-Nyquist sampling is given in Fig. \ref{fig:p21} and Fig. \ref{fig:p22}, respectively.
   In  the sub-Nyquist regime, we randomly choose $K = 250$ Fourier coefficients in each PRI, leading to a  compression ratio of $50\%$. 
   The hit rate of matrix OMP in the on-grid case is also displayed for comparison. 
   From Fig. \ref{fig:p21} and Fig. \ref{fig:p22}, matrix OMP exhibits a serious performance degradation in the off-grid case compared to the on-grid case, if no over-discretization is performed.
   It it observed that the hit rate is only around $0.7$ even when the SNR is high enough because of the mismatch of observation model.
   Nevertheless, the performance degradation can be significantly relieved by over-discretization.
   When $\gamma \geq 2$, performance of matrix OMP in the off-grid case is still worse than the counterpart in the on-grid case, but the performance gap is not significant, especially when $\gamma$ is large.
   For a high SNR, the performance loss can be effectively reduced by increasing $\gamma $. 
   The results here indicate that matrix OMP is still applicable in our problem by properly decreasing the grid size.
}

 \color{black}
      
    \begin{figure}
    	\centering
    	\includegraphics[width=\linewidth]{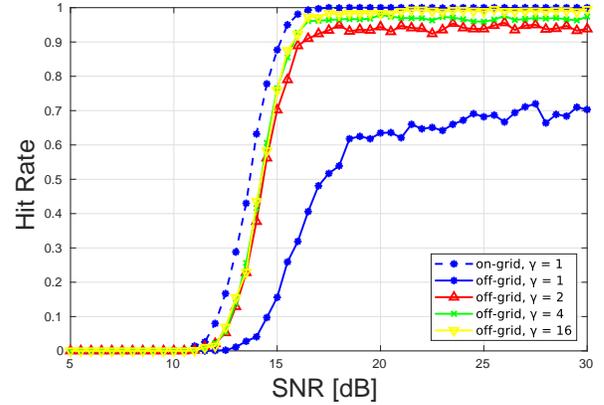}
    	\caption{Delay-Doppler recovery performance for RPPC in the on-grid and off-grid case, where Nyquist sampling is performed.}
    	\label{fig:p21}
    \end{figure}

  \begin{figure}
	\centering
	\includegraphics[width=\linewidth]{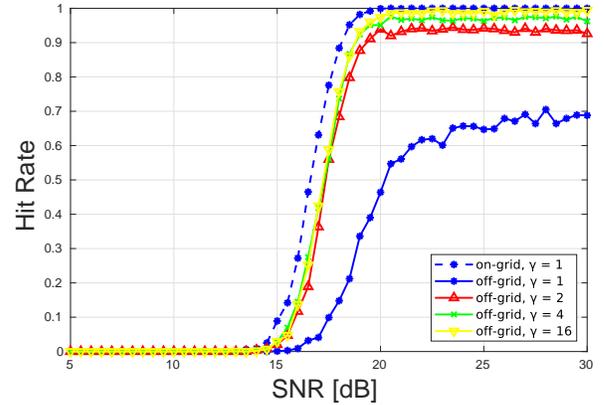}
	\caption{Delay-Doppler recovery performance for RPPC in the on-grid and off-grid case, where sub-Nyquist sampling with compression ratio of $50\%$  is performed.}
	\label{fig:p22}
\end{figure}

   \subsection{Impact of  number of targets}
   
    \color{black}
   
   \Copy{KeySimu1}
  {We also performed simulations to examine the impact of number of targets on the recovery performance.
  Specifically, the hit rate versus SNR is calculated  and demonstrated in Fig. \ref{fig:p3} for $L = 7,9,11$, in both Nyquist and sub-Nyquist regimes, where $Q=4$ and the targets lie in the center of range-Doppler resolution bins.
  In  the sub-Nyquist regime, we randomly choose $K = 250$ Fourier coefficients in each PRI, leading to a  compression ratio of $50\%$.
  It is observed that the recovery performance only slightly decreases as $L$ increases for both Nyquist and sub-Nyquist sampling if $L$ is not very large.
  Note that from the recovery condition obtained in Sec. \ref{sec-condition1} for sub-Nyquist sampling, perfect recovery is guaranteed for arbitrary sparse matrix $\X$ with $L < 9$.
  That means if $L < 9$, the targets can always be recovered in the noise-less case, regardless of their locations and RCSs, while if $L \geq 9$, the targets may not be correctly reconstructed, depending on their parameters.
  Nevertheless, as shown in Fig. \ref{fig:p3}, the targets can still be recovered with a high probability when the targets are uniformly distributed, even if $ L \geq 9 $, namely the sparse recovery condition is not met.
  This result suggests that the sparsity constraint to radar targets  can be relaxed in practical use, extending the application scope of sub-Nyquist sampling.}
  
    \begin{figure}
  	\centering
  	\includegraphics[width=\linewidth]{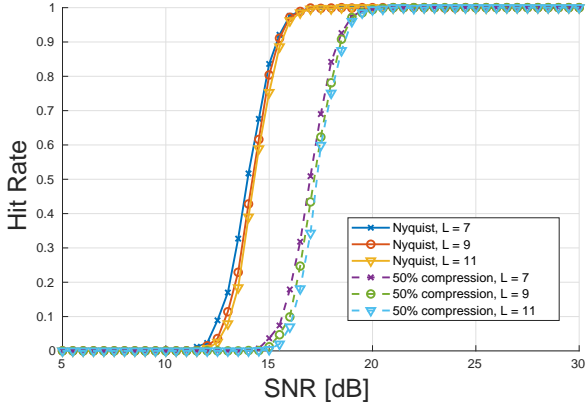}
  	\caption{Delay-Doppler recovery performance for $L = 7,9,11$, in both Nyquist and sub-Nyquist regimes.}
  	\label{fig:p3}
  \end{figure}
   
   \Copy{KeySimu2}
  {When the recovery condition is not met, although the ranges and Dopplers of randomly distributed targets may still be recovered with high probability, there should exist some radar target scenes in which recovery fails. 
  Moreover, even if the recovery condition is met, the targets may not be perfectly recovered with practical algorithms like OMP and $\ell_1$ norm minimization. 
  To show this,  consider the worst case in which all the $L$ targets are located in the same reduced range resolution bin, i.e.
  \begin{displaymath}
  	n_0 = \cdots = n_{L-1} .
  \end{displaymath} 
  The velocities are uniformly distributed in the unambiguous region $[0, 1/T_r)$.
  The ambiguity orders can be arbitrary integers in $[0,Q-1]$.
  The delays and Dopplers of targets lie at the center of the corresponding resolution bins.
  
  From the recovery condition, if $L < (P-Q+2)/2$, the targets can be  recovered with probability 1 by finding the sparsest solution of  \eqref{eq-inverse_problem}.
  However, as practical algorithms may not find the sparsest solution, the hit rate for them may be less than $1$. 
  Here, we evaluate the impact of the number of users on the hit rate under this target scene for $\ell_1$ norm minimization and OMP, under random pulse phase and random locations of targets.
  The simulation is run in the noiseless case and in the Nyquist regime. 
  In particular, we recover $\bm X$ by solving the sub-problems in \eqref{eq-ch3-recovery-problem-4} via $\ell_1$ norm minimization and OMP.
  To reduce the computational complexity, we first detect the reduced range resolution bin where the targets lie, and then  solve the sub-problem in the detected reduced range resolution bin.
  
  The hit rate of versus number of users for $\ell_1$ norm minimization   is given in Fig. \ref{fig:p41}, for different $Q $ and $P$.
  In Fig. \ref{fig:p41}, if $Q  = 1$, namely there is no range ambiguity, the hit rate is always  1 regardless of the number of targets, since the equation in \eqref{eq-inverse_problem} well-determined.
  If $Q > 1$, the hit rate is close to 1 for small $L$, i.e. range ambiguity can be resolved for sparse targets.
  When $L$ becomes larger, the targets are not sparse enough and the hit rate can be rather low. 
  From Fig. \ref{fig:p41}, it is observed that the number of recoverable targets can be increased by transmitting more pulses or reducing the ambiguity order.
  This observation is also verified by the theoretical bound $(P-Q+2)/2$. 
  It is also observed that, to achieve a hit rate close to 1, the   maximal number of recoverable targets with $\ell_1$ norm minimization is less than the  theoretical bound, indicating the performance gap between $\ell_1$ norm minimization and $\ell_0$  minimization.
  
  The hit rate of versus number of users for OMP   is given in Fig. \ref{fig:p42}, for different $Q $ and $P$.
  Comparing Fig. \ref{fig:p41} and \ref{fig:p42}, the recovery performance of OMP is worse than $\ell_1$ norm minimization although OMP generally has lower computation load. Nevertheless, OMP still guartantees a high hit rate when the number of targets is small.
}
    
     
     \color{black}
    

    \begin{figure}
   	\centering
   	\includegraphics[width=\linewidth]{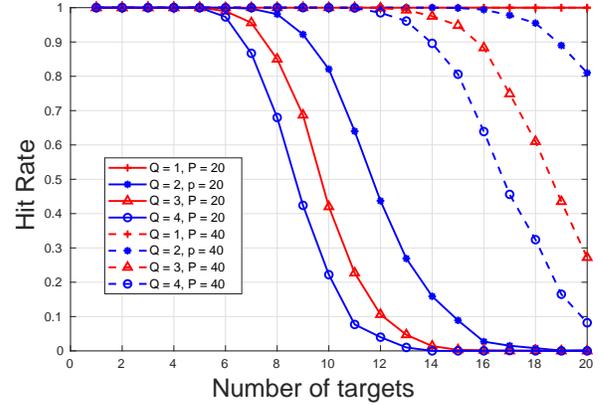}
   	\caption{Hit rate versus number of targets in the Nyquist regime for $\ell_1$ norm minimization.}
   	\label{fig:p41}
   \end{figure}

     \begin{figure}
   	\centering
   	\includegraphics[width=\linewidth]{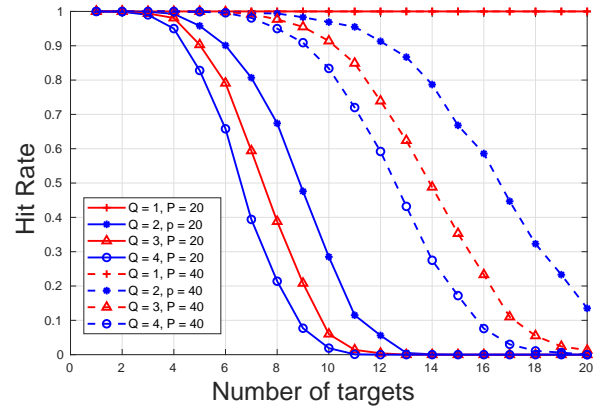}
   	\caption{Hit rate versus number of targets in the Nyquist regime for  OMP.}
   	\label{fig:p42}
   \end{figure}

   
      \subsection{Computation time of matrix OMP}
      
       \color{black}
      
      \Copy{CompTime}{
      In the end, we compare the computation time of matrix OMP under different number of targets.
      This comparison is performed for different $P$, in both Nyquist and sub-Nyquist regimes. 
      In sub-Nyquist sampling,  we randomly choose $K = 250$ Fourier coefficients, leading to a compression ratio of $50\%$.
      The average computation time versus $L$ is shown in Fig. \ref{fig:p5}.
      From  Fig. \ref{fig:p5}, we observe that the computation time of matrix OMP approximately grows linearly with  $L$ when $L$ is small.
      When more pulses are transmitted, namely $P$ becomes larger, the computation load becomes higher.
      The computation time in the sub-Nuquist regime is less than the counterpart in the Nyquist regime, since sub-Nyquist sampling reduces the size of data.}
      
       \color{black}

       \begin{figure}
      	\centering
      	\includegraphics[width=\linewidth]{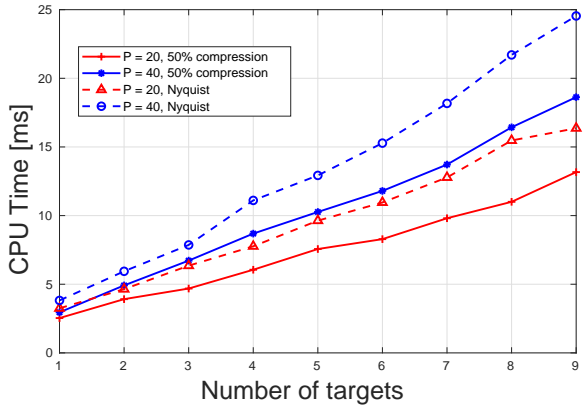}
      	\caption{Computation time versus number of targets $L$ in both Nyquist and sub-Nyquist regimes.}
      	\label{fig:p5}
      \end{figure}
 
    \section{Conclusion}  \label{section6}
    
     \color{black}
    
   \Copy{KeyConclusion}{In this paper, a random pulse phase coding approach is proposed to resolve the range ambiguity of pulse-Doppler radars.
   The advantage of our approach is that the samples from all pulses can be jointly processed to estimate the range-Doppler parameters, and thus the SNR is improved compared to the MPRF method.
   For random pulse phase coding, we establish a range-Doppler recovery problem, which is under-determined.  
   To unambiguously recover the ranges and Dopplers, we propose to solve this problem by sparse recovery algorithms, which can be used in both Nyquist and sub-Nyquist regimes.
   We analyze the performance of sparse recovery by deriving the maximal number of recoverable targets in the noiseless case, given the number of samples $K$ in each PRI, the number of transmit pulses $P$ and the maximal ambiguity order $Q$.
   In particular,  sparse recovery guarantees unambiguous recovery if the number of targets in each reduced range resolution bin is less than $(P-Q+2)/2$ in the Nyquist regime. 
   In the sub-Nyquist regime, the recovery condition is tighter, and requires that the total number of targets is less than $\min \{ (K+1)/2,(P-Q+2)/2 \}$.
   Simulations demonstrate that our approach outperforms  MPRF  in terms of detection rate, in both Nyquist and sub-Nyquist regimes.
   We also verified that sparse recovery algorithms, like matrix OMP, are still applicable even when the delays and Dopplers of the targets do not lie in the center of  resolution bins.}
    
  \Copy{KeyLimitation1}{Despite the above contributions, our approach still has some limitations.
  First,  our approach is proposed for slowly-fluctuating targets. If the target is fast-fluctuating, the RCS of a target varies from pulse to pulse, which corrupts the coded pulse phase.}
  \Copy{KeyLimitation2}{Second, the proposed target reconstruction method requires that the targets satisfy certain sparse conditions. When the target scene is not sparse, range ambiguity is not guaranteed to be correctly resolved.
  Finally, although our recovery method is applicable by reducing the grid size in the off-grid case, this strategy increases the computation complexity and still exhibits performance loss due to model mismatch.
  More effective sparse recovery algorithms should be applied in the off-grid case.}

     \color{black}
    
    \appendix[Proof of THEOREM \ref{Theorem::A}]
    
    Under the assumption that $ \phi[p] $ is independently and uniformly distributed in $[0, 2 \pi)$, we will prove that $\beta_B = P-Q+2$ with probability one.
    Since $\beta_B \leq P-Q+2$, we  need to prove that  $\beta_B  \geq P-Q+2$, namely any $P-Q+1$ columns of $\B$ are linearly independent. 
    
    Let $\overline{\B}$ be the matrix consisting of the $(Q-1)$-th  to the $(P-1)$-th rows in $\B$.
    In the rest of the proof, we  prove  that any $P-Q+1$ columns of $\overline{\B}$ are linearly independent with probability one. As a corollary, any  $P-Q+1$ columns of $\B$ are linearly independent  with probability one.
    
    For convenience, we use the following notations. Let $u = P - Q + 1$.
    Define
    \begin{equation}
        \mathcal{C} = \left\{ (c_0,c_1,\ldots,c_{u-1})  \ \middle| \ 0 \leq c_0 < c_1 < \cdots < c_{u-1}  \leq PQ-1  \right\},
    \end{equation}
    which consists of all the $u$-combinations of the column index set $\{ 0, \ldots,PQ-1 \}$ for $\overline{\B}$.
    For each $\c = (c_0,\ldots,c_{u-1}) \in \mathcal{C} $, we stack the $  c_0$-th, $  c_1$-th,..., and the $c_{u-1} $-th  columns of $\overline{\B}$ into a square matrix $\D (\c)$.
    Let $c_v = P q_v  + p_v$, where $q_v$ is an integer in $[0,Q-1]$, $p_v$ is an integer in $[0,P-1]$, and $q_ 0 \leq q_1 \leq \cdots \leq q_{u-1}$, for  $v = 0,\ldots, u-1$.
    Then the $(b-Q+1,v)$-th entry of $\D(\c)$ is given by $ W _ {P} ^ {b p_v } z[ b-q_v ] $, for $b = Q - 1 ,\ldots, P-1$ and  $ v = 0 , \ldots , u-1$.

    The columns of a square matrix are linearly independent if and only if the determinant of the matrix is not zero.
    Therefore, the statement that any $P-Q+1$ columns of $\overline{\B}$ are linearly independent is equivalent to any of the following two statements:
    \begin{enumerate}[(1)]
        \item For any $\c \in \mathcal{C} $, $f(\z; \c) = \det (\D(\c)) \neq 0$.
        \item $F(\z) = \Pi_{\c \in \mathcal{C}} f(\z;\c)\neq 0$.
    \end{enumerate}
    Here, $\z = (z[0],\ldots,z[P-1])= \left(  e ^ {j \phi[0]},\ldots, e ^ {j \phi[P-1]} \right)$, for  $ p = 0 ,\ldots, P-1$. 
    
    We note that both $ F(\z)$ and $ f(\z; \c) $ can be expressed as a polynomial with respect to $\z$.  
    To prove that any $P-Q+1$ columns of $\overline{\B}$ are linearly independent with probability one, we only need to prove that $F (\z) \neq 0$ with probability one.    
    
    To finish the proof, we  first prove that $f(\z;
    \c)$ is a nonzero polynomial, as stated in Lemma \ref{lemma::B}. As an immediate consequence of Lemma \ref{lemma::B}, $F(\z)$ should be a nonzero polynomial. 
    Later in Lemma \ref{Lemma::A}, we  use a strengthened version of the well-known fact that the Harr measure of the  zeros of a nonzero polynomial is zero.
    Combining Lemma \ref{lemma::B} and Lemma \ref{Lemma::A}, since $F(\z)$ is a nonzero polynomial, the Harr measure of  its zeros is zero.
    Therefore, the probability for $F(\z) = 0$ is zero, and  $F(\z) \neq 0$  with probability one.
    
    We start the proof by presenting Lemma~\ref{lemma::B}.
    
    \begin{mylemma} \label{lemma::B}
        For any $\c$ in $\mathcal{C}$, $  f(\z;\c) = \det(\D(\c))$ is a nonzero polynomial, namely there exists    $\z \in \mathbb{C}^P$ such that $  f(\z;\c) \neq 0$.
    \end{mylemma}
    \begin{proof}
        We note that the $ v$-th columns of $\D(\c)$ is selected from the  $q_v$-th block of $\overline{\B}$, for $v = 0,\ldots,u-1$.
        We  first prove the proposition under the condition that all columns of $\D(\c)$ are from the same block in $\overline{\B}$, and then prove the proposition under the opposite condition.
        Correspondingly, the proof is divided into the two following situations:
        \begin{enumerate}[(a)]
            \item $q_0 = q_1 = \ldots = q_{u - 1}$;
            \item $q_0,q_1,\ldots,q_{u-1}$ are not all the same.
        \end{enumerate}
        
        \emph{Proof for situation (a)}: 
        In this case, the  $(b-Q+1,v)$-th entry of $\D(\c)$ is $ W _ {P} ^ {b p_v } z[ b-q_0 ]$ and $p_0 < p_1 < \cdots < p_{u-1}$. 
       To calculate $f(\z; \c)$, the determinant of $\D(\c)$, we apply the following operations to $\D(\c)$: First divide the $(b-Q+1)$-th row by $ z[b-q_0]$ and then divide the $v$-th column by $ W_{P}^{p_v(Q-1)} $, for $b = Q-1, \ldots, P-1$ and $ v = 0, \ldots, u-1$. 
        These operations result in a new matrix, denoted by $\hat{\D} (\c)$, whose $(b',v)$-th entry is given by $[\hat{\D} (\c)]_{b',v} = W_p^{b' p_v}$, for $b'=0,\ldots,u-1$ and $v = 0,\ldots,u-1$.
       According to the property of determinant \cite{linear_algebra}, one has 
        \begin{equation}
           f(\z;\c) =   \det \big(\hat{\D} (\c) \big) \prod_{b = Q-1}^{P-1} z[b-q_0] \prod _ {v=0} ^ {u-1}  W_{P}^{p_v(Q-1)}.
       \end{equation}
        We note that $\hat{\D}(\c) $ is a Vandermonde matrix and its determinant is not zero since its bases $W_{P}^{p_0},\ldots,W_{P}^{p_{u-1}}$ are distinct \cite{alexeev_full_2012}.
        Therefore, $f(\z;\c)$ is certainly a nonzero polynomial. 
        
        \emph{Proof for situation (b)}: Since the values of $q$ are not identical, we have $Q \geq 2$. The proposition can be proved by mathematical induction for situation (b). 
        
        First, we check the proposition when $\D(\c) \in \mathbb{C}^{u \times u}$ has a low dimension, say, $u = 2$. In this case, one has $Q = P-1$, $\c = (c_0, c_1)$, $  P-2 \leq b \leq P-1 $, and we let $q_0 < q_1$. Then, $\D(\c)$ is written as 
        \begin{displaymath}
            \D(\c) = \left[
            \begin{array}{cc}
                W_P^{(P-2) p_0} z[P-2-q_0] & W_P^{(P-2) p_1} z[P-2-q_1] \\ 
                W_P^{(P-1) p_0} z[P-1-q_0] &  W_P^{(P-1) p_1} z[P-1-q_1]
            \end{array} 
            \right],
        \end{displaymath} 
        with its determinant  $f(\z;\c)$ given by 
        \begin{displaymath}
            \begin{aligned}
                f(\z;\c) &= W_P^{(P-2) p_0} W_P^{(P-1) p_1} z[P-1-q_1]  z[P-2-q_0] \\ &- W_P^{(P-1) p_0}   W_P^{(P-2) p_1} z[ P-2-q_1 ] z[ P-1-q_0 ].
            \end{aligned}
        \end{displaymath} 
    It can easily verified that $f(\z;\c)$ is a non-zero polynomial when $q_0 < q_1$. 
        
        Next, suppose that the proposition holds for $u = 2, 3,\ldots, u' -1$, or equivalently $ Q= P-1,\ldots, Q'+1$ , where $Q' = P - u' + 1$ and $2 \leq Q' \leq P-2$.
        We need to prove that it holds for $u = u'$ or $Q = Q'$.
        To this aim, we recall that $ q_ 0 \leq q_1 \cdots \leq q_{ u ' -1} $ because $c_ 0 < c_1 < \cdots < c_{u'-1}$.
        Since $q_0,q_1,\ldots,q_{u'-1}$ are not all the same, there exists an integer $t$ satisfying that $q_0 = \cdots = q_{t-1}$ and $q_{t-1} < q_{t}$, where $1 \leq t \leq u'-1$. 

       With the Leibniz formula  \cite{linear_algebra}, the determinant of $\D(\c)$ is expressed as
        \begin{equation} \label{det1}
            \begin{aligned}
                f(\z;\c) &= \sum_{\v \in \mathcal{S}} \mathrm{sgn} (\v) \prod _ {i=0} ^ {u'-1}  \D_{i, \v_i} (\c) \\
                &=\sum_{\v \in \mathcal{S}} \mathrm{sgn} (\v) \prod _ {i=0} ^ {P-Q'}  W_{P} ^ {(i+Q'-1)p_{\v_i}}  z [i + Q' - 1 - q_{\v_i}],
            \end{aligned}
        \end{equation}
        where the sum is computed over all the permutations $\v$ for $\{ 0, 1,\ldots, P-Q'\}$, and $\mathcal{S}$ is the set consisting of all such permutations.
        For each permutation $\v$, $\mathrm{sgn} (\v) $ is the sign of $\v$.
        If $\v$ can be obtained by interchanging two elements in $(0,1,\ldots,P-Q')$ for an even number of times, $\mathrm{sgn} (\v) = 1$. Otherwise,  $\mathrm{sgn} (\v) = -1$.
        
        We note that $f(\z;\c)$ is expressed as the sum of several monomials.
        These monomials can be divided into two groups, according to whether the monomial includes the variables $ z[P-1-q_0], \ldots,  z[P-t-q_0] $.
        We denote the sum of the monomials which include $ z[P-1-q_0], \ldots,  z[P-t-q_0] $ by $f_1(\z;\c)$, and denote the  sum of the monomials which do not include them or only include part of them by $f_2(\z;\c)$.
        It is clear that $f(\z;\c) = f_1(\z;\c) + f_2(\z;\c)$.
        Hereinafter, we show that $f_1(\z;\c)$ is a nonzero polynomial.
        If  $f_1(\z;\c)$ is a nonzero polynomial, $ f(\z;\c)$ should be a nonzero polynomial.
        Otherwise, one has $f_2(\z;\c) = -  f_1(\z;\c)$, indicating that the monomials in $f_2(\z;\c)$ include $ z[P-1-q_0], \ldots,  z[P-t-q_0] $, which is opposite to the definition of $f_2(\z;\c)$.

        
        For $ i < u ' -t $, it holds that
        \begin{equation}
            i + Q' - 1 - q_{\bm v_i} < u ' -t + Q'-1 - q_0 = P-t-q_0.
        \end{equation}
        Therefore, $ \prod _ {i=0} ^ {u'-t-1}  \D_{i, \v_i} (\c)$ does not include $ z[P-1-q_0], \ldots,  z[P-t-q_0] $.
        For $\v \in \mathcal{S}$, if $ \prod _ {i=0} ^ {u'-1}  \D_{i, \v_i} (\c)$ includes $ z[P-1-q_0], \ldots,  z[P-t-q_0] $, $ \prod _ {i=u'-t} ^ {u'-1}  \D_{i, \v_i} (\c)$ should include $ z[P-1-q_0], \ldots,  z[P-t-q_0] $, indicating that
        \begin{equation}
            i + Q' - 1 - q_{\bm v_i} = i + Q' - 1 - q_{0} \ \Rightarrow \ q_{\bm v_i} =  q_{0},
        \end{equation}  
        and further
        \begin{equation} \label{eq-vi}
            0 \leq \bm v_i \leq t-1,
        \end{equation}
        for $\ u' - t \leq i \leq u ' -1$.
        
        %
        Let $ \v = [\v^{(1)},\v^{(2)}] $, where $ \v^{(1)} =  ( \v_{0},\ldots,\v_{u'-t-1} )  $ and $ \v^{(2)} =  ( \v_{u'-t},\ldots,\v_{u'-1} ) $ .
        From  \eqref{eq-vi}, we have $ \v^{(1)}  \in \mathcal{S}_1 $  and $ \v^{(2)}  \in \mathcal{S}_2 $, where 	 $\mathcal{S}_1$ is the set of all the permutations of $\{ t, \ldots, u'-1\}$, and  $\mathcal{S}_2$ is the set of all the permutations of  $\{ 0, \ldots, t-1\}$.
        Then $f_1(\z;\c)$ can be expressed as 
        \begin{subequations}
            \begin{align}
                &	f_1(\z;\c) \notag \\
                =	&	\sum_{\v^{(1)} \in \mathcal{S}_1} \sum_{\v^{(2)} \in \mathcal{S}_2} \mathrm{sgn} (\v ) \prod _ {i=0} ^ {u'-1}  \D_{i,  \v _ i} (\c) \\
                =	& \sum_{\v^{(1)} \in \mathcal{S}_1} \sum_{\v^{(2)} \in \mathcal{S}_2} \mathrm{sgn} (\v ) \prod _ {i_1=0} ^ {u'-t - 1} \prod _ {i_2=u'-t} ^ {u'-1}   \D_{i_1,  \v _ {i_1}} (\c)  \D_{i_2,  \v _ {i_2}} (\c) \\
                =	&	\sum_{\v^{(1)} \in \mathcal{S}_1} \sum_{\v^{(2)} \in \mathcal{S}_2} \mathrm{sgn} (\v^{(1)})  \mathrm{sgn} (\v^{(2)}) \notag \\ 
                & \times \prod _ {i_1=0} ^ {u'-t-1} \prod _ {i_2=u'-t} ^ {u'-1}  \D_{i_1, \v^{(1)}_{i_1}} (\c) \D_{i_2, \v^{(2)}_{i_2+t-u'}} (\c) \label{eq-hg-3} \\
                =	&	\sum_{\v^{(1)} \in \mathcal{S}_1} \sum_{\v^{(2)} \in \mathcal{S}_2} \mathrm{sgn} (\v^{(1)})  \mathrm{sgn} (\v^{(2)}) \notag \\
                &  \times \prod _ {i_1=0} ^ {u'-t-1} \prod _ {i_2=0} ^ {t-1}  \D_{i_1, \v^{(1)}_{i_1}} (\c) \D_{i_2 + u' -t , \v^{(2)}_{i_2}} (\c). \label{eq-hg-4}
            \end{align}
        \end{subequations}
        Here, \eqref{eq-hg-3} comes from the fact that
        \begin{displaymath}
            \begin{aligned}
                & \mathrm{sgn} (\v ) = \mathrm{sgn} (\v^{(1)})  \mathrm{sgn} (\v^{(2)}), \\
                & \v _ {i_1} = \v^{(1)}_{i_1}, \ 0 \leq i_1 \leq u'-t-1, \\
                & \v _ {i_2} = \v^{(2)}_{i_2+t-u'}, \ u'-t \leq i_2 \leq u'-1,
            \end{aligned}
        \end{displaymath} 
        and \eqref{eq-hg-4} is obtained via replacing $i_2$ with $i_2 + u' - t$.

        To further explore the property of $h_1(\z;\c) $,  we  express $\D(\c)$ in block matrix form
        \begin{equation}
            \D(\c) = \left[
            \begin{array}{cc}
                \vdots  &  \E(\c)  \\ 
                \G(\c) &  \cdots
            \end{array}       
            \right],
        \end{equation}
        where $\G(\c)$ is a $t \times t$ matrix and  $\E(\c)$ is a $(u'-t) \times (u'-t)$ matrix. 
        According to the relationship that 
        \begin{displaymath}
            \begin{aligned}
                &	\D_{i_1, \v^{(1)}_{i_1}}(\c) = \E_{i_1, \v^{(1)}_{i_1}-t} (\c), \ 0 \leq i_1 \leq u'-t-1, \\
                &	\D_{i_2 + u'-t, \v^{(2)}_{i_2}}(\c) = \G_{i_2, \v^{(2)}_{i_2}} (\c), \ 0 \leq i_2 \leq t-1,
            \end{aligned}
        \end{displaymath}		
        we have
        \begin{subequations}
            \begin{align}
                f_1(\z;\c)
                =&  \sum_{\v^{(1)} \in \mathcal{S}_1} \sum_{\v^{(2)} \in \mathcal{S}_2} \mathrm{sgn} (\v^{(1)}) \mathrm{sgn} (\v^{(2)}) \notag \\ & \times \prod _ {i_1=0} ^ {u'-t-1} \prod _ {i_2=0} ^ {t-1}  \E_{i_1, \v^{(1)}_{i_1}-t} (\c) \G_{i_2, \v^{(2)}_{i_2}} (\c)\\ 
                =&  \ \bigg(  \sum_{\v^{(1)} \in \mathcal{S}_1}  \mathrm{sgn} (\v^{(1)}) \prod _ {i_1=0} ^ {u'-t-1} \E_{i_1, \v^{(1)}_{i_1}-t} (\c)\bigg) \notag \\
                &\times  \bigg(  \sum_{\v^{(2)} \in \mathcal{S}_2}  \mathrm{sgn} (\v^{(2)})\prod _ {i_2=0} ^ {t-1}\G_{i_2, \v^{(2)}_{i_2}} (\c)\bigg) \\
                =&  \det(\G(\c)) \det(\E(\c)).
            \end{align}
        \end{subequations}
        
        \color{black}
         The matrix $\E(\c)$ has a similar expression with $\D(\c)$ and the $(b-Q'-t+1, v)$-th entry of $\E(\c)$ is $W _ {P} ^ {(b-t) p_{v+t} } z[ b-t-q_{v+t} ]$, for $ b = Q'+t-1, \ldots, P-1$ and $ v = 0, \ldots, u'-t-1$.
        Let $\hat{q}_v = t + q_{v+t}$ and $\hat{p}_v = p_{v+t} $.
        Multiply the $v$-th column of $\E(\c)$ by $W_P ^ {t \hat{p} _ {v}} $, for $v = 0 ,\ldots, u'-t-1$, and the result is a new matrix $\hat{\E} (\c)$ whose $(b-(Q'+t)+1, v)$-th entry is given by $W _ {P} ^ {b \hat{p}_{v} } z[ b-\hat{q}_{v} ]$, for $ b = (Q'+t)-1, \ldots,  P-1$ and $ v = 0 ,\ldots, u'-t-1$.
        One can observe that $\hat{\E} (\c)$ has a consistent expression with $\D(\c)$ while its dimension is less than $\D(\c)$.
        Using the induction hypothesis for $Q = Q'+t$, the determinant of $\hat{\E}(\c)$ is a nonzero polynomial of $\z$.
         \color{black}
        The $(b,v)$-th entry of $\G(\c)$ is given by $W_{P} ^ {-b p_v} z[b - q_0]$, for $b = P-t ,\ldots, P - 1$ and $ v = 0,\ldots, t-1$.
        It is clear that $\G(\c)$ has the same structure as $\D(\c)$ in the proof for situation (a).
        Similarly, one can prove that $\det(\G(\c))$ is a nonzero polynomial.
        Since $\det(\G(\c))$ and $ \det(E(\c))$ are nonzero polynomials, $f_1(\z;\c)$ is a  nonzero polynomial, and  $f(\z;\c)$ is also a  nonzero polynomial.


        Now that we have proved that the proposition holds for $Q=Q'$.
        With mathematical induction, the proposition holds for all $Q = 2 ,\ldots, P-1$ under situation (b).
        In conclusion, $f(\z;\c)$ is a nonzero polynomial under both situations (a) and (b), and the proof is complete. 
    \end{proof}

    Lemma \ref{lemma::B} shows that $f(\z;\c) $ is a nonzero polynomial for all $ \c \in \mathcal{C}$.
    Therefore, $F(\z) = \Pi_{\c \in \mathcal{C}} f(\z;\c) $ is a nonzero polynomial.
    We  use the following lemma from \cite{8494717} to prove that $F(\z) \neq 0$ with probability one, which points out the fact that the Harr measure of the set composed of zeros for a nonzero polynomial is zero.

    \begin{mylemma} \label{Lemma::A}
        Let $F(z_0,\ldots,z_{P-1})$ be a nonzero complex polynomial with $P$ variables.
        Define the set of zeros of $F$ 
        \begin{equation} \label{eq-Nf}
            \mathcal{N}_F = \left\{ (z_0,\ldots,z_{P-1}) \in \mathbb{C}^P \ \middle| \ F(z_0,\ldots,z_{P-1}) = 0 \right\}.
        \end{equation}
        Define the $P$-torus
        \begin{equation} \label{eq-appendix-torus}
            \mathcal{T}^P = \underbrace{ \mathcal{T}^1  \times \cdots \times \mathcal{T}^1}_{P} , \ \mathcal{T}^1 = \left\{ z \in \mathbb{C} \ \middle| \ |z| = 1 \right\},
        \end{equation}
        where $\times$ denotes the Cartesian product.
        Let $\sigma^P$ be the Harr measure on $\mathcal{T}^P$.
        Then one has
        \begin{equation}
            \sigma^P (\mathcal{T}^P \cap \mathcal{N}_F ) = 0.
        \end{equation}
        
    \end{mylemma}

    Define the map $\Theta : [0,2 \pi)^P \rightarrow \mathcal{T}^P$ by $\Theta(\phi_0,\ldots,\phi_{P-1}) = \left( e^{-j \phi_0},\ldots, e^{-j \phi_{P-1}} \right)$.
    The map $ \Theta $  is bijective and absolutely continuous.
    Let $\mu ^ P$ be the Harr measure on $ [0,2 \pi) ^ P $.
    Because $ \sigma ^ P (\mathcal{T} ^ P \cap \mathcal{N}_F) = 0$, one has
    \begin{equation} \label{equation::mesaure::2}
        \mu ^ P ( \Theta ^ {-1} ( \mathcal{T}^{P} \cap \mathcal{N}_F ) )  = 0.
    \end{equation}
    If $\phi[0], \ldots, \phi[P-1]$ are independent and uniformly distributed in $[0, 2 \pi)$, the probability of $F(\z) = 0$ is
    \begin{equation}
        \frac{   \mu ^ P ( \Theta ^ {-1} ( \mathcal{T}^{P} \cap \mathcal{N}_F ) ) }{\mu ^ P ( [0, 2\pi )^P)} = \frac{   \mu ^ P ( \Theta ^ {-1} ( \mathcal{T}^{P} \cap \mathcal{N}_F ) ) }{ (2 \pi)^P} = 0.
    \end{equation}
    In other words, $F(\z) \neq 0$ with probability one.
    Further, any $P-Q+1$ columns of $\B$ are linearly independent with probability one, completing the proof.

    \bibliographystyle{IEEEtran}
    \bibliography{refs}

\end{document}